%% file: main.tex
\documentclass[11pt]{article}

\usepackage[]{agmacros}
\usepackage{andrews_bib}
\title{An Improved Construction of Variety-Evasive Subspace Families}
\author{Robert Andrews\thanks{Cheriton School of Computer Science, University of Waterloo. Email: randrews@uwaterloo.ca.} \and Abhibhav Garg\thanks{Cheriton School of Computer Science, University of Waterloo. Email: abhibhav.garg@uwaterloo.ca.}}
\date{}

\begin{document}

\maketitle

\begin{abstract}
  We study the question of explicitly constructing variety-evasive subspace families, a pseudorandom primitive introduced by Guo (Computational Complexity 2024) that generalizes both hitting sets and lossless rank condensers.
  Roughly speaking, a variety-evasive subspace family $\cH$ is a collection of subspaces such that for every algebraic variety $V$ in a fixed family $\cF$, there is some subspace $W \in \cH$ that is in general position with respect to $V$.

  We give an explicit construction of a subspace families that evade all degree-$d$ varieties in an $n$-dimensional affine or projective space.
  Our construction improves on the size of the variety-evasive subspace families constructed by Guo and, for varieties of degree $n^{1 + \Omega(1)}$, comes within a polynomial factor of Guo's lower bound on the size of any such variety-evasive subspace family.
  Our variety-evasive subspace families rely on an improved construction of hitting sets for Chow forms of algebraic varieties.
\end{abstract}

\section{Introduction} \label{section: introduction}
\input{sections/introduction}

\section{Preliminaries} \label{section: preliminaries}
\input{sections/preliminaries}

\section{Construction of variety-evasive subspace families} \label{section: construction}
\input{sections/construction}
\printbibliography

\end{document}

%% file: sections/introduction.tex
\subsection{Background}

The probabilistic method is a powerful tool with widespread use throughout mathematics.
Many objects of interest throughout combinatorics and computer science, such as Ramsey graphs and randomness extractors, can be shown to exist via the probabilistic method, often with excellent parameters.
While probabilistic arguments excel at demonstrating existence, they fall short of providing an efficient, deterministic means of constructing the object of interest.
One of the main goals of pseudorandomness is to obtain explicit constructions of objects that exhibit properties similar to those enjoyed by a randomly-chosen object.
Not only are explicit constructions interesting in their own right, but they are all the more important in algorithmic applications such as derandomization, where the aim is to reduce or eliminate the use of randomness as an algorithmic resource.

Over time, the theory of pseudorandomness has developed into a rich area of theoretical computer science that studies a variety of objects, including expander graphs, list-decodable codes, randomness extractors, pseudorandom generators, and the numerous connections and relationships between them \cite{Vadhan12}.
More recently, a parallel theory of linear-algebraic pseudorandomness has also appeared.
This theory studies collections of linear maps that satisfy various pseudorandom properties, including rank extractors and rank condensers, subspace designs, dimension expanders, and subspace-evasive sets, where the dimension of vector spaces plays a role analogous to that of min-entropy in classical pseudorandomness.

In this work, we study the construction of \emph{variety-evasive subspace families}, a linear-algebraic pseudorandom object introduced by \textcite{Guo24}.
Roughly speaking, a variety-evasive subspace family $\cH$ is a collection of linear or affine spaces such that for any algebraic variety $V$ of interest, there is some subspace $W \in \cH$ that is in general position with respect to $V$.
For example, if $V$ is a curve in the plane and $\cH$ is a collection of lines, then there should be some line $W \in \cH$ such that the intersection $V \cap W$ is nonempty and zero-dimensional.
More generally, when the variety $V$ is fixed, most linear spaces (and in fact, most varieties) $W$ will satisfy $\codim{V \cap W} = \codim{V} + \codim{W}$, and a variety-evasive subspace family must contain a subspace $W$ that intersects $V$ in a manner similar to a random subspace.
This is the sort of behavior captured by the notion of one variety evading another, defined below.

\begin{definition}[Evasiveness]
  Let $V$ and $W$ be irreducible subvarieties of $\bP^n$ or $\bA^n$.
  We say that $W$ \emph{evades} $V$ if
  \[
    \dim(V \cap W) \le \dim(V) + \dim(W) - n.
  \]
  If $V$ and $W$ are affine varieties, we say that $W$ \emph{strongly evades} $V$ if the above holds with equality whenever the right-hand side is nonnegative.
  If $V$ is a reducible variety, we say that $W$ \emph{(strongly) evades} $V$ if $W$ (strongly) evades each irreducible component of $V$.
\end{definition}

When $V$ and $W$ are affine varieties, it may be the case that $V \cap W = \varnothing$, even though $\codim{V} + \codim{W}$ is not large enough to force this intersection to be empty in general.
The notion of $W$ strongly evading $V$ adds the constraint that the intersection $V \cap W$ should have the expected codimension of $\codim{V} + \codim{W}$ whenever this is possible.

As a matter of notational convenience, we also define the notion of a $k$-subspace family.

\begin{definition}[$k$-Subspace family]
  For $0 \le k \le n$, a \emph{$k$-subspace family} is a finite collection of $k$-dimensional subspaces of $\bP^n$.
  Similarly, an \emph{affine $k$-subspace family} is a finite collection of $k$-dimensional affine subspaces of $\bA^n$.
\end{definition}

We can now define variety-evasive subspace families, the main object of study in this work.

\begin{definition}[Variety-evasive subspace families \cite{Guo24}]
  Let $\cF$ be a family of subvarieties of $\bP^n$ (or of $\bA^n$) and let $\cH$ be a $k$-subspace (or affine $k$-subspace) family.
  \begin{enumerate}
    \item 
      We say that the family $\cH$ is \emph{(strongly) $\cF$-evasive} if for every $V \in \cF$, there is some $W \in \cH$ that (strongly) evades $V$.
      In the case where $\cF$ is the family of all degree-$d$ subvarieties, we say that $\cH$ is \emph{(strongly) $(n,d)$-evasive}.
    \item
      We say that $\cH$ is \emph{(strongly) $(\cF, \eps)$-evasive} if for every $V \in \cF$, a randomly-chosen $W \in \cH$ (strongly) evades $V$ with probability at least $1 - \eps$.
      When $\cF$ is the family of all degree-$d$ subvarieties, we say that $\cH$ is \emph{(strongly) $(n,d,\eps)$-evasive}.
      \qedhere
  \end{enumerate}
\end{definition}

As observed by \textcite{Guo24}, variety-evasive subspace families are a common generalization of two natural problems in algebraic pseudorandomness.
The first such problem is the construction of hitting sets for arithmetic circuits, a common approach to derandomizing algorithms for the polynomial identity testing problem.
A hitting set $\cH \subseteq \bF^n$ is a collection of points such that for every nonzero polynomial $f \in \cC$ in a circuit class $\cC$ of interest, there is some point $\valpha \in \cH$ such that $f(\valpha) \neq 0$.
The point $\valpha$ witnesses the fact that $f$ is nonzero as a polynomial.
An efficient, deterministic construction of such a hitting set immediately implies a deterministic algorithm with similar complexity for testing polynomials $f \in \cC$.
Geometrically, one can view a hitting set as a set of points $\cH$ such that for every hypersurface $V$ defined by a polynomial $f \in \cC$, there is a point $\valpha \in \cH$ such that $\valpha \notin V$.
Taking $\cF_\cC$ to be the family of hypersurfaces defined by polynomials in $\cC$, a hitting set is precisely a $\cF_\cC$-evasive $0$-subspace family.

The second task generalized by variety-evasive subspace families is the construction of lossless rank condensers.
Rank condensers, first appearing in the work of \textcite{GR08}, are another of the myriad objects studied in linear-algebraic pseudorandomness.
A lossless rank condenser $\cE \subseteq \bF^{r \times n}$ is a collection of linear maps $\bF^n \to \bF^r$ such that for every $r$-dimensional subspace $V$, there is some map $E \in \cE$ where the image of $V$ under $E$ remains $r$-dimensional.
Dualizing, this is precisely the same as requiring the kernel of $E$ to only intersect $V$ at the origin, or equivalently, that the linear spaces in projective space $\bP^{n-1}$ corresponding to $\ker E$ and $V$ have empty intersection.
Thus, a lossless rank condenser is exactly the same object as an $(n-r-1)$-subspace family in $\bP^{n-1}$ that evades all dimension-$(r-1)$ linear spaces.

Besides their generalization of other tasks in algebraic pseudorandomness, explicitly constructing variety-evasive subspace families is a natural task from the point of view of algebraic geometry.
For example, as we will see later, variety-evasive subspace families can easily be used to construct Noether normailizing maps for algebraic varieties.

\textcite{Guo24} gives the only known construction of variety-evasive subspace families.
To describe his result, we need the following function.
For $n, d \in \bN$ and $k \in \bc{0,1,\ldots,n-1}$, let $N(k,d,n)$ be the function given by
\[
  N(k,d,n) \coloneqq \min\br{\binom{(k'+1)(n'+1+d)}{(k'+1)d}, \binom{(n-k)(n'+1+d)}{(n-k)d}}
\]
where $k' \coloneqq \min(k, d-2)$ and $n' \coloneqq n - (k - k')$.
\textcite[Theorem 1.7]{Guo24} constructed $(n,d,\eps)$-evasive $k$-subspace families of size 
\[
  \poly\br{N(k,d,n), n, \eps^{-1}}.
\]
Moreover, this construction is explicit and can be carried out in time polynomial in the size of the variety-evasive subspace family $\cH$ (and $\log p$ if the field $\bF$ has characteristic $p > 0$).
Guo's work complements this with a lower bound that shows any $(n, d)$-evasive $k$-subspace family $\cH$ must have size at least
\[
  |\cH| \ge \begin{cases}
    (n-k)(k+1)+1 & \text{if $d = 1$,} \\
    \max\br{d(n-k)(k+1)+1, \binom{d+n-k}{d} + (n-k+1)k} & \text{if $d > 1$.}
  \end{cases}
\]
This lower bound is tight in the sense that there is a non-explicit construction of a variety-evasive $k$-subspace family whose size precisely matches the lower bound \cite[Section 4]{Guo24}.

\subsection{Our results}

We give an improved construction of $(n,d,\eps)$-evasive $k$-subspace families.
Just like Guo's construction, ours is explicit and can be carried out in time polynomial in the size of the evasive subspace family $\cH$ (and $\log p$, if $\bF$ has characteristic $p > 0$).

\begin{restatable}{theorem}{maintheorem} \label{theorem: main}
  Let $n, d \in \bN$, $k \in \bc{0,1,\ldots,n-1}$, and $0 < \eps < 1$.
  There is an explicit $(n,d,\eps)$-evasive $k$-subspace family (resp.\ strongly $(n,d,\eps)$-evasive affine $k$-subspace family) $\cH$ of size $\poly(((n-k)d + 1)^{n-k}, \eps^{-1})$.
  Moreover, there is a deterministic algorithm that on input $(n, d, \eps)$ runs in time $\poly(|\cH|)$ and prints defining equations of the subspaces in $\cH$.
\end{restatable}

To compare our construction with that of Guo, we use the simplified bound from \cite[Theorem 1.7]{Guo24}, which asserts an explicit construction of an $(n,d,\eps)$-evasive $k$-subspace family of size $\poly(n^{\min(k+1, n-k, d) d}, d, \eps^{-1})$.
In contrast, we construct explicit $(n,d,\eps)$-evasive $k$-subspace families of size $\poly(((n-k)d+1)^{n-k}, \eps^{-1})$.
This improves the dependence on $d$, which may be exponentially larger than $n$ and $k$; for example, a variety cut out by $n/2$ quadratic equations will generally have degree $2^{n/2}$.

We also give a simple construction of $(n,d)$-evasive $k$-subspace families that eliminates the polynomial overhead appearing in \cref{theorem: main}.

\begin{restatable}{theorem}{simpleconstruction} \label{theorem: easy construction}
  Let $n, d \in \bN$ and $k \in \bc{0,1,\ldots,n-1}$.
  There is an explicit $(n,d)$-evasive $k$-subspace family (resp.\ strongly $(n,d)$-evasive affine $k$-subspace family) $\cH$ of size at most $(nd + 1)^{n-k}$.
  Moreover, there is a deterministic algorithm that receives $n$, $d$, and $k$ as input and prints the subspace family $\cH$ in time $(nd+1)^{O(n-k)}$.
\end{restatable}

When $d \ge n^{1 + \eps}$ for some constant $\eps > 0$, \cref{theorem: easy construction} comes within a polynomial factor of Guo's lower bound on the size of $(n,d)$-evasive $k$-subspace families.
To see this, observe that we can bound the size of the $(n,d)$-evasive $k$-subspace family constructed in \cref{theorem: easy construction} by $d^{O(n-k)}$.
On the other hand, Guo showed that any such $k$-subspace family must have size at least
\[
  \binom{d+n-k}{n-k} \ge \br{1 + \frac{d}{n-k}}^{n-k} \ge \br{1 + \frac{d}{n}}^{n-k} \ge d^{\Omega(n-k)}.
\]

As an application of our main result, we give an explicit construction of Noether normalizing maps.
Noether normalization is a useful basic result in commutative algebra.
In geometric language, the Noether normalization lemma states that for every $r$-dimensional affine variety $V \subseteq \bA^n$, there is a surjective finite morphism $\pi : V \to \bA^r$.
Finite morphisms are particularly nice maps of algebraic varieties: every subvariety of $V$ is mapped to a subvariety of $\bA^r$, and the preimage $\pi^{-1}(p)$ of any point $p \in \bA^r$ will be a finite set.

Over an infinite field, a randomly-chosen linear map $\pi : \bA^n \to \bA^r$ will be Noether normalizing with high probability.
\textcite[Theorems 1.12 and 1.13]{Guo24} observed that Noether normalizing maps can easily be constructed from an $(n,d)$-evasive subspace family.
As a corollary of \cref{theorem: main}, we obtain the following explicit construction of Noether normalizing maps.

\begin{corollary} \label{cor: noether normalize}
  Let $n, d \in \bN$, $r \in \bc{0,1,\ldots,n}$, $k \coloneqq n - r - 1$, and let $\eps > 0$.
  There is an explicit collection $\cL$ of linear maps $\bA^n \to \bA^r$ of size $\poly(((n-k)d)^{n-k}, \eps^{-1})$ and computable in time $\poly(\abs{\cL})$ such that the following hold.
  \begin{enumerate}
    \item 
      For every projective variety $V \subseteq \bP^n$ of dimension $r-1$ and degree at most $d$, at least a $1-\eps$ fraction of maps $\pi \in \cL$ induce a surjective finite morphism from $V$ to $\bP^{r-1}$.
    \item 
      For every affine variety $V \subseteq \bA^n$ of dimension $r$ and degree at most $d$, at least a $1-\eps$ fraction of maps $\pi \in \cL$ restrict to a surjective finite morphism from $V$ to $\bP^r$.
  \end{enumerate}
\end{corollary}

The proof of \cref{cor: noether normalize} is exactly the same as \cite[Theorems 1.12 and 1.13]{Guo24}, but using the $(n,d,\eps)$-evasive subspace families of \cref{theorem: main} instead of the subspace families constructed by Guo, so we omit further details.

%% file: sections/preliminaries.tex
Throughout this work, we use $\bF$ to denote an algebraically closed field.
The ring of polynomials in $n+1$ variables is denoted $\bF\bs{x_{0}, \dots, x_{n}}$.
We use bold letters for vectors, whose length will be clear from context.
We use vector notation to denote monomials: given $\va \in \bN^{n+1}$, we write $\vx^{\va}$ for the monomial $x_{0}^{a_{0}} \cdots x_{n}^{a_{n}}$.
We use $\bA^{n}$ and $\bP^{n}$ to denote $n$-dimensional affine and projective spaces, respectively.

\subsection{Algebraic geometry} \label{subsection: algebraic geoemtry}

We assume basic algebraic geometry at the level of \cite{shafarevich1994basic1}.
For us, an \emph{algebraic variety} is a Zariski-closed subset of $\bA^n$ or $\bP^n$; in particular, a variety may be reducible.
We adopt the convention that the degree of a variety is the sum of the degrees of all of its irreducible components.
For this notion of degree, the following inequality holds, which we call Bézout's inequality.
\begin{lemma}[B\'{e}zout's inequality]
  \label{lemma: bezout}
  For any two varieties $V$ and $W$, we have $\deg(V \cap W) \leq \deg V \cdot \deg W$.
\end{lemma}
A proof of the affine version of \cref{lemma: bezout} is given in \cite[Theorem~1]{Heintz83}, while a proof of the projective version follows immediately from \cite[Chapter~2.3]{fulton1984introduction}.

We also need the following lemma, which shows that varieties of small degree and dimension are contained in linear subspaces of small dimension.
\begin{lemma}
  \label{lemma: small span}
  For every projective variety $V$, there exists a linear subspace $L$ of dimension at most $\dim V \cdot \deg V$ such that $V \subseteq L$.
\end{lemma}
\begin{proof}
  First, suppose $V$ is irreducible.
  Let $L \subseteq \bP^{n}$ be a minimal linear subspace such that $V \subseteq L$.
  By \cite[Corollary~18.12]{harris2013algebraic}, we can deduce that $V$ has codimension at most $\deg V - 1$ within $L$.
  It follows that $L$ has dimension at most $\dim V + \deg V - 1$.

  Now let $V$ be an arbitrary projective variety and let $V_{1}, \dots, V_{t}$ be the irreducible components of $V$.
  Each $V_{i}$ is contained in a linear subspace $L_{i}$ of dimension at most $\dim V_{i} + \deg V_{i} - 1$.
  Let $L$ be the smallest linear subspace that contains $\bigcup_{i=1}^{t} L_{i}$.
  Since each irreducible component $V_i$ is contained in $L$, the variety $V$ is likewise contained in $L$.
  We have 
  \[
    \dim L \leq \sum_{i=1}^{t} \dim L_{i} \leq \sum_{i=1}^{t} (\dim V_{i} + \deg V_{i} - 1).
  \]
  By Bézout's inequality, we have $\sum_{i=1}^{t} \deg V_{i} \leq \deg V$ and $t \leq \deg V$.
  Combined with the fact that $\dim V_{i} \leq \dim V$, the required statement follows.
\end{proof}

\subsection{Variety-evasive subspace families}

In this subsection, we gather some preparatory lemmas for the construction of variety-evasive $k$-subspace families.
In particular, we will see that to construct $(n,d,\eps)$-evasive $k$-subspace families, it suffices to construct a $k$-subspace family that evades degree-$d$ varieties of dimension $n - k - 1$, and that we may further assume these varieties live in an ambient space of dimension $(n-k-1) d$.
In addition to this, we will see that an $(n,d,\eps)$-evasive $k$-subspace family automatically yields an $(n,d,\eps')$-evasive affine $k$-subspace family with only a mild loss in $\eps'$.

We first cite a lemma of \textcite{Guo24} that shows constructing an $(n,d,\eps)$-evasive $k$-subspace family reduces to the task of constructing a $k$-subspace that $\eps$-evades equidimensional subvarieties of dimension $n - k - 1$ and degree $d$.

\begin{lemma}[{\cite[Lemma~3.1]{Guo24}}]
  \label{lemma: reduction to codim k}
  Let $\cF$ be the family of equidimensional projective subvarieties of $\bP^n$ of dimension $n - k - 1$ and degree at most $d$.
  Suppose $\cH$ is an $(\cF, \eps)$-evasive $k$-subspace family.
  Then $\cH$ is $(n, d, \varepsilon)$-evasive.
\end{lemma}

Next, we show that the dimension of the ambient space can be reduced.
As a consequence of \cref{lemma: small span}, it suffices to construct variety-evasive subspace families when the ambient dimension $n$ is at most $\deg V \cdot \dim V$.
Essentially the same reduction is used in the constructions of \textcite{Guo24}.
We give a proof for completeness.

\begin{lemma}
  \label{lemma: ambient reduction}
  Let $r \coloneqq n - k - 1$ and let $\cH$ be a given $(rd-r-1)$-subspace family of linear subspaces of $\bP^{rd}$ of size $T$.
  \begin{enumerate}
    \item 
      If $\cH$ is $(rd, d)$-evasive, then we can construct an $(n, d)$-evasive $k$-subspace family of size $\poly\br{T, n, d}$ in time $\poly\br{T, n, d}$.
    \item
      If $\cH$ is $(rd, d, \eps)$-evasive, then we can construct an $(n, d, 2 \eps)$-evasive $k$-subspace family of size $\poly\br{T, n, d, \eps^{-1}}$ in  time $\poly\br{T,n,d,\eps^{-1}}$.
  \end{enumerate}
\end{lemma}
\begin{proof}
  We first carry out the construction of an $(n,d)$-evasive $k$-subspace family in detail.
  After this, we briefly remark how to modify the construction to obtain an $(n,d,\eps)$-evasive $k$-subspace family.

  We begin by constructing an $(n,1)$-evasive $(n-rd-1)$-subspace family $\widehat{\cH}$.
  By \cref{lemma: reduction to codim k}, it suffices to build a collection $W_{1}, \dots, W_{s}$ of linear subspaces in $\bP^{n}$ such that for any linear subspace $L \subset \bP^{n}$ of dimension $rd$, there is some $W_{i}$ such that $L \cap W_{i} = \emptyset$.
  Such a family of subspaces $W_i$ is precisely the nullspaces of the matrices in a \emph{rank extractor}.
  \textcite{FS12}, together with the improved analysis of \textcite{FSS14}, gave an explicit construction of such a rank extractor of size $s = rd(n-rd)+1$.
  For each $i \in [s]$, we pick points $v_{i, 0}, \dots, v_{i, n} \in \bP^{n}$ such that $\lspan{v_{i, 0}, \dots, v_{i, n}} = \bP^{n}$ and $\lspan{v_{i, rd+1}, \dots, v_{i, n}} = W_{i}$.

  Let $\cH = \bc{Y_{1}, \dots, Y_{T}} \subseteq \bP^{rd}$ be the given $(rd, d)$-evasive $(rd-r-1)$-subspace family.
  For each $Y_{i}$, we pick $rd - r$ points $u_{i, 0}, \dots, u_{i, rd - r - 1} \in \bP^{rd}$ that span $Y_{i}$.

  Now we define linear subspaces $Z_{i, j} \subseteq \bP^{n}$ for every $1 \leq i \leq s$ and $1 \leq j \leq T$ that will form the required $(n, d)$-evasive $k$-subspace family.
  Let $\phi_{i} : \bP^{rd} \to \bP^n$ be the linear map that sends the point $e_{\ell}$ to $v_{i, \ell}$, where $e_{\ell}$ denotes the point in $\bP^{rd}$ whose $\ell$-th coordinate is $1$ and whose remaining coordinates are zero.
  We take $Z_{i,j}$ to be the $k$-subspace given by
  \[
    Z_{i,j} \coloneqq \lspan{v_{i,rd+1}, \ldots, v_{i,n}, \phi_i(u_{j,0}), \ldots, \phi_i(u_{j, rd-r-1})}.
  \]
  Given the subspace families $\cH$ and $\widehat{\cH}$, it is clear that we can construct the subspaces $Z_{i,j}$ in $\poly(T,n,d)$ time by iterating over the product set $\cH \times \widehat{\cH}$.

  It remains to show that the $k$-subspace family $\bc{Z_{i,j} : i \in [s], j \in [T]}$ is $(n, d)$-evasive.
  To this end, let $V \subseteq \bP^{n}$ be an equidimensional variety of dimension $r$ and degree $d$.
  By \cref{lemma: small span}, there is a linear subspace $X$ of dimension $rd$ that contains $V$.
  Since $\widehat{\cH} = \bc{W_{1}, \dots, W_{s}}$ is $(n, 1)$-evasive, there is some $W_i$ such that $W_i \cap X = \varnothing$; without loss of generality, we may assume that $W_{1} \cap X = \emptyset$.
  For each $0 \leq i \leq rd$, the intersection $\lspan{W_{1}, v_{i, 1}} \cap X$ consists of a single point in $\bP^{n}$, which we denote by $a_{i}$.
  Comparing dimensions, we see that $\lspan{a_{0}, \dots, a_{rd}} = X$.
  Let $\psi: \bP^{rd} \to \bP^{n}$ denote the linear map that maps $e_{i}$ to $a_{i}$.
  Since $v_{1, i} \in \lspan{W_{1}, \psi\br{e_{i}}}$, it follows that $\lspan{W_{1}, \psi(e_{i})} = \lspan{W_{1}, \phi_{1}(e_{i})} = \lspan{W_{1}, v_{i, 1}}$.

  The image of the map $\psi$ is exactly $X$, and so $\psi$ is an isomorphism between $\bP^{rd}$ and $X$.
  Since $V \subseteq X$, the preimage of $V$ under $\psi$ is also a variety of dimension $r$ and degree $d$.
  Because $\cH = \bc{Y_{1}, \dots, Y_{T}}$ is $(rd, d)$-evasive, there is some $Y_j$ such that $Y_j \cap \psi^{-1}(V) = \varnothing$; without loss of generality, we may assume that $Y_{1}$ does not intersect $\psi^{-1}(V)$.
  The choice of the map $\psi$ implies that $Z_{1, 1} = \lspan{W_{1}, \psi(Y_{1})}$.
  Now suppose towards a contradiction that there is a point $b \in V \cap Z_{1, 1}$.
  The point $b$ lies in $\lspan{\psi(c), w}$, for some $c \in Y_{1}$ and $w \in W_{1}$.
  Since $b, \psi(c) \in X$ and $X \cap W_{1} = \emptyset$, we deduce $b = \psi(c)$ in $\bP^{n}$.
  This contradicts the fact that $Y_{1} \cap \psi^{-1}(V) = \emptyset$.

  Since $V$ was an arbitrary equidimensional variety, we can invoke \cref{lemma: reduction to codim k} to conclude that the constructed $k$-subspace is $(n, d)$-evasive.
  To construct an $(n,d,\eps)$-evasive $k$-subspace family, the only part of the construction that changes is that the $(n,1)$-evasive subspace family used in the first step must be replaced by an $(n,1,\eps)$-evasive subspace family.
  The Forbes--Shpilka construction of rank extractors also produces this $(n,1,\eps)$-evasive subspace family, incurring only a $1/\eps$ multiplicative increase in the size of the subspace family.
\end{proof}

Finally, we show that any $(n,d,\eps)$-evasive $k$-subspace family immediately yields an $(n,d,\eps')$-evasive affine $k$-subspace family with only a small loss in the parameter $\eps'$.

\begin{lemma}
  \label{lemma: strong affine evasive}
  Let $\cH$ be an $(n,d,\eps)$-evasive $k$-subspace family.
  Let $\cH' \subseteq \cH$ be the subfamily of subspaces not contained in the hyperplane at infinity.
  Then the restriction of $\cH'$ to the affine chart of points at finite distance is a strongly $(n, d, 2 \varepsilon / (1 - \varepsilon))$-evasive affine $k$-subspace family.
\end{lemma}
\begin{proof}
  Since $\cH$ is $(n, d, \varepsilon)$-evasive, at most $\varepsilon$ fraction of the elements of $\cH$ are contained in $\var{x_{0}}$, the hyperplane at infinity, so $\abs{\cH'} \geq (1 - \varepsilon) \abs{\cH}$.
  For the rest of this proof, let $\varepsilon' \coloneqq \varepsilon / (1 - \varepsilon)$.
  The family $\cH'$ is $(n, d, \varepsilon')$-evasive.

  Let $\cH_{\infty}$ be the restriction to $\var{x_{0}}$ of those subspaces in $\cH'$.
  This is a $(k-1)$-subspace family of size at least $(1-\varepsilon) \abs{\cH}$.
  Let $V \subseteq \bP^{n-1}$ be any variety of degree at most $d$.
  We can identify $\bP^{n-1}$ with the hyperplane at infinity in $\bP^{n}$.
  This way, $V$ can be viewed as a variety in $\bP^{n}$.
  All except $\varepsilon \abs{\cH}$ elements in $\cH'$ evade $V$.
  Therefore, $\cH_{\infty}$ is a $(n-1, d, \varepsilon')$-variety evasive subspace family.

  We now show that the restriction of $\cH'$ to the affine chart is strongly $(n, d, 2 \varepsilon')$-evasive for affine varieties.
  Let $V \subseteq \bA^{n}$ be an affine variety of degree $d$.
  Identifying $\bA^{n}$ with the affine chart, we can consider $V \subseteq \bP^{n}$.
  Let $\overline{V}$ be the projective closure of $V$, and let $V_{\infty}$ denote $\overline{V} \cap \var{x_{0}}$.
  Both $\overline{V}$ and $V_{\infty}$ have degree at most $d$.

  Suppose $H \in \cH'$ is such that $H$ evades $\overline{V}$ and $H \cap \var{x_{0} = 0}$ evades $V_{\infty}$.
  Let $W$ be an irreducible component of $V$, let $\overline{W}$ be the projective closure of $W$, and let $W_{\infty}$ denote $\overline{W} \cap \var{x_0}$.
  By assumption, every component of $\overline{W} \cap H$ has dimension $\dim \overline{W} + k - n$.
  If $\dim \overline{W} + k - n < 0$, then this intersection is empty.
  In this case, $W \cap H \cap \bA^{n}$ is also empty, and therefore $H \cap \bA^{n}$ evades $W$.
  Otherwise, we have $\dim \overline{W} + k - n \geq 0$.

  To show that $H \cap \bA^{n}$ evades $W$, it suffices to show that $\overline{W} \cap H \not\subseteq \var{x_{0}}$.
  This will guarantee that $W \cap H \cap \bA^{n}$ is nonempty and has the correct dimension.
  To this end, note that $\dim W_{\infty} = \dim W - 1$.
  Since $H \cap \var{x_{0}}$ evades $W_{\infty}$, every component of $\overline{W} \cap H \cap \var{x_{0}}$ has dimension $\dim \overline{W} + k - n - 1$.
  Since every component of $\overline{W} \cap H$ has dimension $\dim \overline{W} + k - n$, no component of $\overline{W} \cap H$ lies completely in the hyperplane at infinity, completing the proof that $H \cap \bA^{n}$ evades $W$.
  Since this holds for every component $W$ of $V$, we deduce that $H \cap \bA^{n}$ evades $V$.

  Because $\cH'$ and $\cH_{\infty}$ are $(n,d,\eps')$-evasive and $(n-1,d,\eps')$-evasive, respectively, the fraction of hyperplanes in $\cH'$ such that either $H$ does not evade $\overline{V}$ or such that $H \cap \var{x_{0}}$ does not evade $V_{\infty}$ is at most $2 \varepsilon'$.
  Therefore, the restriction of $\cH'$ to the affine chart is a strongly $(n, d, 2\varepsilon')$-evasive affine $k$-subspace family.
\end{proof}

%% file: sections/construction.tex
We give two constructions of variety-evasive subspace families.
The first is a simple construction meant to illustrate our main technical observation.
The second is a more involved construction that allows us to obtain improved epsilon variety-evasive subspace families.

\subsection{Basic construction} \label{subsection: basic construction}

Let $V \subseteq \bP^{n}$ be an projective variety of dimension $r$.
The intersection of $V$ with a general hyperplane is a variety of dimension $r-1$.
Typically, to obtain an effective version of this fact, one picks a point on each component of $V$ and controls the probability that a randomly-chosen hyperplane $H$ does not pass through the chosen points.
If $H$ does not pass through these points, then $H$ properly intersects each component of $V$, so the intersection $V \cap H$ will have dimension exactly $r-1$.

We implement this argument in the following lemma.
Here, the coefficients of the hyperplane $H$ are not picked independently, but are instead chosen from points of the form $(1, \gamma, \gamma^2, \ldots, \gamma^n)$ for $\gamma \in \bF$.

\begin{lemma}
  \label{lemma: small parameter intersection}
  Let $V \subseteq \bP^{n}$ be a projective variety of dimension $r$ and degree $d$.
  Let $V_{1}, \dots, V_{t}$ be the irreducible components of $V$.
  There exists a polynomial $P \in \bF\bs{z}$ of degree at most $n d$ such that for any $\gamma$ with $P(\gamma) \neq 0$, every hyperplane section $V_{i} \cap \var{\sum_{j=0}^{n}\gamma^{j} x_{j}}$ has dimension $\dim V_{i} - 1$.
\end{lemma}

\begin{proof}[Proof of \cref{lemma: small parameter intersection}]
  From each irreducible component $V_{i}$ of $V$, we pick a point $p_{i}$.
  The number of components of $V$ is bounded by the degree $d$.
  Consider the polynomial $Q(\vx,z) \coloneqq \sum_{i=0}^n z^{i} x_{i}$.
  The evaluation $Q(p_i, z)$ is a nonzero polynomial in $z$, since the point $p_i$ has at least one nonzero coordinate.
  Set $P(z) \coloneqq \prod_{i = 1}^t Q(p_{i}, z)$.
  
  It is clear that $P$ is a polynomial of degree $n t \le nd$.
  To show that $P$ satisfies the claimed property, let $\gamma \in \bF$ be a point such that $P(\gamma) \neq 0$ and consider the linear form $L(\vx) \coloneqq \sum_{i=0}^n \gamma^i x_i$.
  By construction, we have that $L(p_i) = Q(p_i, \gamma)$.
  Because $P(\gamma) \neq 0$, it follows that $Q(p_i, \gamma) \neq 0$.
  In particular, the linear form $L$ does not vanish identically $V_i$.
  Therefore $V_{i} \cap \var{L}$ has dimension $\dim V_{i} - 1$ by \cite[Chapter~9, Section~4, Corollary~4]{CLO15}.
\end{proof}

We can also deduce a similar statement for affine varieties.

\begin{lemma}
  \label{lemma: small parameter intersection affine}
  Let $V \subseteq \bA^{n}$ be an affine variety of degree $d$.
  Let $V_{1}, \dots, V_{t}$ be the irreducible components of $V$.
  There exists a polynomial $P \in \bF\bs{z}$ of degree at most $n d$ such that for any $\gamma$ with $P(\gamma) \neq 0$, every hyperplane section $V_{i} \cap \var{1 + \sum_{j=1}^{n}\gamma^{j} x_{j}}$ has dimension $\dim V_{i} - 1$.
\end{lemma}

\begin{proof}[Proof of \cref{lemma: small parameter intersection}]
  Let $\overline{V}$ denote the projective closure of $V$.
  The number of components of $\overline{V}$ is bounded by the degree $d$.
  Each irreducible component $\overline{V}_{i}$ of $\overline{V}$ is the projective closure of a component $V_{i}$ of $V$.
  Therefore, for every $i$ we have $\dim (\overline{V}_{i} \cap \var{x_{0}}) = \dim \overline{V}_{i} - 1$.
  For each $i$ such that $\dim \overline{V}_{i} \geq 1$, and for each irreducible component $\overline{V}_{i, j}$ of $\overline{V}_{i} \cap \var{x_{0}}$, pick a point $p_{i, j} \in \overline{V}_{i,j}$.
  For those $i$ such that $\dim \overline{V}_{i} = 0$, let $p_{i, 1}$ be the single point of $\overline{V}_{i}$.

  Consider the polynomial $Q(\vx, z) \coloneqq \sum_{i=0}^n z^{i} x_{i}$.
  If we evaluate at $p_{i, j}$, we get $Q(p_{i, j}, z)$, which is a nonzero polynomial in $\bF\bs{z}$, since $p_{i, j}$ has at least one nonzero coordinate.
  Set $P(z) \coloneqq \prod_{i, j} Q(p_{i}, z)$.
  If we pick $\gamma \in \bF$ such that $P(\gamma) \neq 0$, then the linear form $L(\vx) \coloneqq \sum_{i=0}^n \gamma^{i} x_{i}$ is nonzero at every $p_{i, j}$.
  From this, we can deduce that $\dim(\overline{V}_{i} \cap \var{L, x_{0}}) = \dim \overline{V}_{i} - 2$ for every irreducible component $\overline{V}_{i}$ of $\overline{V}$ of dimension at least one.
  Therefore, for any such component, we must have $\dim (V_{i} \cap \var{\ell}) = \dim V_{i} - 1$, where $\ell(\vx) \coloneqq L(1, x_{1}, \dots, x_{n})$.
  For the zero-dimensional components of $V$, we can also deduce that $V_{i} \cap \var{\ell} = \emptyset$, since $L$ does not vanish at these components.
\end{proof}

\cref{lemma: small parameter intersection} allows us to construct a set $\cH$ of $nd+1$ hyperplanes such that for any projective variety $V \subseteq \bP^n$ of degree $d$, there is some $H \in \cH$ that properly intersects $V$, and \cref{lemma: small parameter intersection affine} does the same for affine varieties.
By iterating this construction $k$ times, we can obtain an $(n,d)$-evasive (affine) $k$-subspace family.
We remark that this solves \cite[Open Problem 7]{Guo24}, which asked for an explicit construction of a polynomial-size $(n,d)$-evasive $k$-subspace family when $n - k = O(1)$.

\simpleconstruction*
\begin{proof}
  We start by establishing the projective version of the theorem.
  Fix pairwise disjoint subsets $B_1, \ldots, B_{n-k} \subseteq \bF$, each of size $nd + 1$ and not containing $0$.
  For each $\alpha \in \bF$, let $L_{\alpha}(\vx)$ denote the linear form $x_{0} + \alpha x_{1} + \cdots + \alpha^{n} x_{n}$.
  For each $\vgamma \in \bF^{n-k}$, let $H_{\vgamma} \subseteq \bP^{n}$ denote the linear subspace $\var{L_{\gamma_{1}}, \dots, L_{\gamma_{n-k}}}$.
  Let $\cH$ denote the collection of subspaces $H_{\vgamma}$ where $\vgamma$ varies over elements in $B_1 \times \cdots \times B_{n-k}$.
  Finally, let $\cH'$ denote the restriction of each element of $\cH$ to the affine chart $x_{0} = 1$.
  We claim that $\cH$ and $\cH'$ are the required $(n,d)$-evasive $k$-subspace and affine $k$-subspace families, respectively.
  The bound on the sizes of $\cH$ and $\cH'$ are clear from construction.
  To algorithmically produce $\cH$ and $\cH'$, we simply iterate over the set $B_1 \times \cdots \times B_{n-k}$ and print the subspace $H_{\vgamma}$ for every $\vgamma \in B_1 \times \cdots \times B_{n-k}$.

  To see why $\cH$ is $(n,d)$-evasive, let $V \subseteq \bP^{n}$ be any variety of degree $d$.
  By \cref{lemma: small parameter intersection}, there exists a polynomial $P_{1} \in \bF\bs{z}$ of degree $nd$ such that for any $\gamma_{1} \in \bF$ satisfying $P_{1}(\gamma_{1}) \neq 0$, the intersection of $V_{i}$ and $\var{L_{\gamma_{1}}}$ has dimension $\dim V_{i} - 1$ for every irreducible component $V_{i}$ of $V$.
  By Bézout's inequality (\cref{lemma: bezout}), the intersection $V \cap \var{L_{\gamma_{1}}}$ has degree at most $d$.
  Since the set $B_1$ has size $nd + 1$, we can find such an element $\gamma_{1} \in B_1$.
  Fix this choice of $\gamma_1$.
  Now invoke \cref{lemma: small parameter intersection} for the variety $V \cap \var{L_{\gamma_{1}}}$.
  This gives us a polynomial $P_{2}$ of degree $nd$ whose non-roots define hyperplanes that properly intersect $V \cap \var{L_{\gamma_1}}$.
  Because $|B_2| = nd + 1$, we can find a nonroot $\gamma_{2}$ of $P_{2}$ that lies in $B_2$.
  Moreover, since $B_1$ and $B_2$ are disjoint, we are guaranteed that $\gamma_2$ is different from $\gamma_{1}$.
  This implies that the intersection $V_i \cap \var{L_{\gamma_1}} \cap \var{L_{\gamma_2}}$ has dimension $\dim V_i - 2$ for every irreducible component $V_i$ of $V$.
  Continuing in this manner, we obtain $(\gamma_{1}, \dots, \gamma_{n-k}) \in B_1 \times \cdots \times B_{n-k}$ such that the subspace $H_{\vgamma}$ evades $V$.

  The affine version for $\cH'$ follows the same argument, invoking \cref{lemma: small parameter intersection affine} in place of \cref{lemma: small parameter intersection}.
  The fact that no $B_i$ contains $0$ implies that each affine subspace in $\cH'$ has dimension $k$.
\end{proof}

\subsection{Constructing $(n,d,\eps)$-evasive subspace families} \label{subsection: full construction}

\cref{theorem: easy construction} constructs an $(n,d)$-evasive $k$-subspace family by building the subspaces iteratively, intersecting $n-k$ hyperplanes one at a time.
To improve this to a construction of an $(n,d,\eps)$-evasive $k$-subspace family, one natural approach is to ensure that at each step, the probability of picking a hyperplane that does not properly intersect a variety $V$ is bounded by $\frac{\eps}{n-k}$.
Applying a union bound over the $n-k$ choices of hyperplanes implies that at most an $\eps$ fraction of subspaces fail to evade $V$.
However, this enlarges the number of hyperplanes chosen at each step by a factor of $\frac{n-k}{\eps}$, which leads to a $(1/\eps)^{n-k}$ factor in the final size of the subspace family.

By arguing about the $n-k$ hyperplanes as a whole instead of one at a time, we can improve this $(1/\eps)^{n-k}$ factor to $(1/\eps)^{O(1)}$.
To do this requires the notion of the Chow form of a variety, which we now define.
The definition and facts we use about the Chow form are classical.
The following statements appear in \cite[Section~2.1.1]{KPS01}, and proofs can be found in \cite[Chapter~1]{philipponindep}.

\begin{definition}
  \label{definition: chow form}
  Let $V \subseteq \bP^{n}$ be an equidimensional variety of dimension $r$ and degree $d$.
  Let $\vu_{0}, \dots, \vu_{r}$ be sets of variables, each of size $n+1$.
  The \emph{Chow form of $V$}, denoted $\cC_{V}$, is a polynomial in $\bF\bs{\vu_{0}, \dots, \vu_{r}}$ with the following properties.
  \begin{itemize}
    \item For every choice of $\vgamma_{0}, \dots, \vgamma_{r} \in \bF^{n+1}$, the linear subspace $\var{\sum_{j=0}^{n} \gamma_{0, j} x_{j}, \dots, \sum_{j=0}^{n} \gamma_{r, j} x_{j}}$ has a point in $V$ if and only if $\cC_{V}(\vgamma_{0}, \dots, \vgamma_{r}) = 0$.
    \item The polynomial $\cC_{V}$ squarefree and homogeneous of degree $d$ in each set of variables $\vu_{i}$. \qedhere
  \end{itemize}
\end{definition}

In other words, the Chow form of a variety characterizes those linear subspaces of dimension at least $n-r+1$ that intersect $V$.
The following is an easy corollary of \cref{lemma: small parameter intersection}.

\begin{corollary} \label{cor: chow form hsg}
  Let $V \subseteq \bP^{n}$ be an equidimensional variety of dimension $r$ and degree $d$ with Chow form $\cC_{V} \in \bF\bs{\vu_{0}, \dots, \vu_{r}}$.
  Let $v_{0}, \dots, v_{r}$ be new variables and let $\phi: \bF\bs{\vu_{0}, \dots, \vu_{r}} \to \bF\bs{\vv}$ be the map defined by $\phi(u_{i, j}) = v_{i}^{j}$.
  Then $\phi\br{\cC_{V}} \neq 0$.
\end{corollary}
\begin{proof}
  Let $\gamma_{0} \in \bF$ be such that $V_{1} \coloneqq V \cap \var{x_{0} + \gamma_{0} x_{1} + \cdots + \gamma_{0}^{n} x_{n}}$ is equidimensional of dimension $r-1$.
  Such an element exists by \cref{lemma: small parameter intersection}.
  The variety $V_{1}$ now is equidimensional of dimension $r-1$, and we can invoke \cref{lemma: small parameter intersection} again to find a $\gamma_{1} \in \bF$ such that $V_{2} \coloneqq V_{1} \cap \var{x_{0} + \gamma_{1} x_{1} + \cdots + \gamma_{1}^{n} x_{n}}$ is equidimensional of dimension $r-2$.
  Repeating this, we find $\gamma_{0}, \dots, \gamma_{r}$ such that $V \cap \var{\sum_{j=0}^{n} \gamma_{0}^{j} x_{j}, \dots, \sum_{j=0}^{n} \gamma_{r}^{j} x_{j}}$ is empty.
  By \cref{definition: chow form}, we have $\phi(\cC_{V})(\gamma_{0}, \dots, \gamma_{r}) \neq 0$, which shows in particular that $\phi(\cC_{V})$ is nonzero.
\end{proof}

In the notation of \cref{cor: chow form hsg}, the polynomial $\phi(\cC_V)$ is a nonzero polynomial of individual degree $nd$ in $r+1$ variables with the property that the subspace defined by the linear forms with coefficients $\vgamma_0, \ldots, \vgamma_r$ evades $V$ if $\phi(\cC_V)(\vgamma_0, \ldots, \vgamma_r) \neq 0$.
The construction of \cref{theorem: easy construction} can be thought of as a construction of a hitting set for the restricted Chow form $\phi(\cC_V)$.
To obtain $(n,d,\eps)$-evasive $k$-subspace families, we instead apply an $\eps$-hitting set to the restricted Chow form $\phi(\cC_V)$.
An \emph{$\eps$-hitting set} has exactly the guarantee we need: it is a finite set $\cH$ such that a $1-\eps$ fraction of points $(\vgamma_0,\ldots,\vgamma_r) \in \cH$ satisfy $\phi(\cC_V)(\vgamma_0,\ldots,\vgamma_r) \neq 0$.
We make this observation precise in the following corollary.

\begin{corollary} \label{cor: hitting set to evasive}
  Suppose $\cP$ is an $\varepsilon$-hitting set for polynomials of individual degree $nd$ in $n-k$ variables.
  Then there exists a $(n, d, \varepsilon)$-evasive $k$-subspace family $\cH$ of size at most $\abs{\cP}$.
  Moreover, there is a deterministic algorithm that takes $\cP$ as input, prints the set $\cH$, and runs in time $\poly(n,\abs{\cP})$.
\end{corollary}
\begin{proof}
  The evasive family of subspaces is constructed as follows.
  For each point $\vmu \in \cP$, let $H_{\vmu}$ be the linear subspace defined by the linear forms $x_{0} + \mu_{i} x_{1} + \cdots + \mu_{i}^{n} x_{n}$ for $i \in [n-k]$. 
  Let $\cH'$ be the resulting family of subspaces.

  Let $V$ an arbitrary equidimensional variety of dimension $n - k - 1$ and degree $d$, and let $\cC_{V}$ be its Chow form.
  Let $\phi$ be the map defined in \cref{cor: chow form hsg}, so $\phi(\cC_{V})$ is a nonzero polynomial of individual degree $nd$ in $n-k$ variables.
  Since $\cP$ is a $\varepsilon$-hitting set for such polynomials, for all except $\varepsilon \abs{\cP}$ elements in $\cP$, we have $\phi(\cC_{V})(\vmu) \neq 0$.
  Therefore, for all except $\varepsilon \abs{\cH'}$ of the subspaces $H \in \cH'$, we have $V \cap H = \emptyset$.

  Note that the exceptional set of $\varepsilon \abs{\cH'}$ vector spaces $H \in \cH'$ such that $V \cap H \neq \emptyset$ contains all those subspaces in $\cH'$ that are not $k$-dimensional.
  Therefore, if we define $\cH$ to be the set of $k$-dimensional subspaces in $\cH'$, then all except $\varepsilon \abs{\cH}$ elements in $\cH$ evade $V$.
  Since $V$ was an arbitrary equidimensional variety of dimension $n-k-1$ and degree $d$, we can invoke \cref{lemma: reduction to codim k} to deduce that $\cH$ is $(n, d, \varepsilon)$-evasive.

  It is clear that the set $\cH'$ can be obtained in polynomial time from the $\eps$-hitting set $\cP$.
  To prune $\cH'$ to the $(n,d,\eps)$-evasive $k$-subspace family $\cH$, we need to remove from $\cH'$ those subspaces that are not $k$-dimensional.
  This is a basic linear algebra calculation that can be performed in $\poly(n)$ time for each element of $\cH'$.
\end{proof}

Finally, we need an explicit construction of $\eps$-hitting sets.
We do this using a construction of \textcite{AGKS15}, as instantiated by \textcite{Guo24}.
The version we state here is adapted to polynomials of bounded individual degree instead of bounded total degree, but follows exactly the same proof as \cite[Lemma 2.3]{Guo24}.

\begin{lemma}[{\cite[Lemma~2.3]{Guo24}, \cite[Lemma~4]{AGKS15}}]
  \label{lemma: epsilon hitting sets}
  For every $d, n \in \bN$ and $\varepsilon \in \br{0, 1}$, there exists an $\varepsilon$-hitting set $\cP$ for the set of all polynomials of individual degree $d$ in $n$ variables.
  The set $\cP$ has size and bit complexity $\poly\br{(d+1)^n, \varepsilon^{-1}}$ and can be constructed in time $\poly(\abs{\cP})$.
\end{lemma}

Combining \cref{cor: hitting set to evasive}, \cref{lemma: epsilon hitting sets}, and the reductions in \cref{section: preliminaries}, we obtain our final construction of $(n,d,\eps)$-evasive subspaces.

\maintheorem*

\begin{proof}
  We start by constructing an $(n,d,\eps)$-evasive $k$-subspace family.
  First suppose $n \leq (n-k)d$.
  Then the required $(n, d, \varepsilon)$-evasive $k$-subspace family can be constructed by using \cref{cor: hitting set to evasive} instantiated with the $\varepsilon$-hitting set from \cref{lemma: epsilon hitting sets}.
  The resulting variety-evasive subspace family has size and bit complexity $\poly\br{(nd+1)^{n-k},\eps^{-1}} \le \poly\br{\br{(n-k)d + 1}^{n-k}, \varepsilon^{-1}}$, since $n \le (n-k)d$.

  Now suppose $n > (n-k) d$.
  Let $r \coloneqq n - k - 1$.
  Using \cref{lemma: epsilon hitting sets} and \cref{cor: hitting set to evasive}, we can construct an $(rd, d, \varepsilon/2)$-evasive $(rd-r-1)$-subspace family.
  This family has size $\poly\br{(rd+1)^r, \eps^{-1}} \le \poly\br{\br{(n-k)d+1}^{n-k}, \varepsilon^{-1}}$.
  We can now invoke \cref{lemma: ambient reduction} to obtain the required $(n, d, \varepsilon)$-evasive subspace family.
  The size of the resulting family is also bounded by $\poly\br{\br{(n-k)d + 1}^{n-k}, \varepsilon^{-1}}$.

  To construct a strongly $(n, d, \varepsilon)$-evasive affine $k$-subspace family, we first pick $\varepsilon'$ such that $\varepsilon = 2 \varepsilon' / (1 - \varepsilon')$.
  We construct a $(n, d, \varepsilon')$-evasive $k$-subspace family using the above construction and then invoke \cref{lemma: strong affine evasive} to obtain the required $(n,d,\eps)$-evasive affine $k$-subspace family.
\end{proof}

%% file: refs.bib
@article{Vadhan12,
  author    = {Salil P. Vadhan},
  title     = {Pseudorandomness},
  journal   = {Foundations and Trends in Theoretical Computer Science},
  volume    = {7},
  number    = {1-3},
  pages     = {1--336},
  year      = {2012},
}

@book{CLO15,
  author    = {David A. Cox and
               John Little and
               Donal O'Shea},
  title     = {Ideals, varieties, and algorithms - an introduction to computational
               algebraic geometry and commutative algebra},
  edition   = {4},
  series    = {Undergraduate texts in mathematics},
  publisher = {Springer},
  year      = {2015},
  isbn      = {978-3-319-16720-6},
  doi = {10.1007/978-3-319-16721-3},
}

@inproceedings{FSS14,
  author    = {Michael A. Forbes and
               Ramprasad Saptharishi and
               Amir Shpilka},
  title     = {Hitting sets for multilinear read-once algebraic branching programs,
               in any order},
  booktitle = {\STOC{2014}},
  pages     = {867--875},
  year      = {2014},
  url       = {https://doi.org/10.1145/2591796.2591816},
  doi       = {10.1145/2591796.2591816},
  bibsource = {dblp computer science bibliography, https://dblp.org}
}

@article{AGKS15,
  author    = {Manindra Agrawal and
               Rohit Gurjar and
               Arpita Korwar and
               Nitin Saxena},
  title     = {Hitting-Sets for {ROABP} and Sum of Set-Multilinear Circuits},
  journal   = {{SIAM} J. Comput.},
  volume    = {44},
  number    = {3},
  pages     = {669--697},
  year      = {2015},
  url       = {https://doi.org/10.1137/140975103},
  doi       = {10.1137/140975103},
  bibsource = {dblp computer science bibliography, https://dblp.org}
}

@inproceedings{FS12,
  author = {Michael A. Forbes and Amir Shpilka},
  title = {On identity testing of tensors, low-rank recovery and compressed sensing},
  year = {2012},
  booktitle = {\STOC{2012}},
  pages = {163--172},
}

@article{GR08,
  author = {Ariel Gabizon and Ran Raz},
  title = {Deterministic extractors for affine sources over large fields},
  journal = {Combinatorica},
  volume = {28},
  pages = {415--440},
  year = {2008},
  doi = {10.1007/s00493-008-2259-3},
}

@article{KPS01,
  author = {Teresa Krick and Luis Miguel Pardo and Mart{\'i}n Sombra},
  title = {{Sharp estimates for the arithmetic Nullstellensatz}},
  volume = {109},
  journal = {Duke Mathematical Journal},
  number = {3},
  publisher = {Duke University Press},
  pages = {521--598},
  year = {2001},
  doi = {10.1215/S0012-7094-01-10934-4},
}

@article{Guo24,
  author = {Guo, Zeyu},
  doi = {10.1007/s00037-024-00256-1},
  isbn = {1420-8954},
  journal = {computational complexity},
  number = {2},
  pages = {10},
  title = {Variety Evasive Subspace Families},
  volume = {33},
  year = {2024},
}

@book{harris2013algebraic,
  title={Algebraic geometry: a first course},
  author={Harris, Joe},
  volume={133},
  year={2013},
  publisher={Springer Science \& Business Media},
  doi = {10.1007/978-1-4757-2189-8},
}

@article{Heintz83,
  author    = {Joos Heintz},
  title     = {{Definability} and fast quantifier elimination in algebraically closed fields},
  journal   = {Theor.~Comput.~Sci.},
  year      = {1983},
  volume    = {24},
  number    = {3},
  pages     = {239--277},
  doi = {10.1016/0304-3975(83)90002-6}
}

@book {shafarevich1994basic1,
    AUTHOR = {Shafarevich, Igor R.},
     TITLE = {Basic Algebraic Geometry 1},
   EDITION = {Second},
 PUBLISHER = {Springer Berlin, Heidelberg},
      YEAR = {1994},
     PAGES = {xx+304},
      ISBN = {3-540-54812-2},
       doi = {10.1007/978-3-642-57908-0},
}

@book{fulton1984introduction,
  title={Introduction to intersection theory in algebraic geometry},
  author={Fulton, William},
  volume={54},
  year={1984},
  publisher={American Mathematical Soc.}
}

@article {philipponindep,
  AUTHOR = {Philippon, Patrice},
  TITLE = {Crit\`eres pour l'ind\'ependance alg\'ebrique},
  JOURNAL = {Inst. Hautes \'Etudes Sci. Publ. Math.},
  FJOURNAL = {Institut des Hautes \'Etudes Scientifiques. Publications
              Math\'ematiques},
  NUMBER = {64},
  YEAR = {1986},
  PAGES = {5--52},
  ISSN = {0073-8301,1618-1913},
  MRCLASS = {11J85},
  MRNUMBER = {876159},
  MRREVIEWER = {Yao\ Chen\ Zhu},
  doi = {10.1007/BF02699191},
}
